   \documentclass[twoside,reqno,11pt]{fcaa-var} %

\usepackage{graphicx}
\usepackage{epsfig}
\usepackage{amsthm}
\usepackage{amsmath}
\usepackage{latexsym}
\usepackage{amsfonts}
\usepackage{amssymb}

\newtheoremstyle{theorem}
  {15pt}          
  {15pt}  
  {\sl}  
  {\parindent}
  {\sc}  
  {. }   
  { }    
  {}     
\theoremstyle{theorem}

\newtheorem{theorem}{Theorem}[section]

\newtheoremstyle{defi}
  {15pt}          
  {15pt}  
  {\rm}  
  {\parindent}     
  {\sc}  
  {. }    
  { }    
  {}     
\theoremstyle{defi}

 \def\proofname{\indent\rm P\ r\ o\ o\ f.}
 
 \def\qed{\hfill$\Box$} 

\usepackage{hyperref} 
\usepackage{enumitem}
\usepackage{dsfont}
\newcommand{\res}{\mathrm{Res}}
\newcommand{\id}{\mathrm{d}}
\newcommand{\ud}{\mathrm{d}}
\newtheorem{proposition}{Proposition}

 \def\theequation{\arabic{section}.\arabic{equation}}

 \def\themycopyrightfootnote{\vspace*{3pt}
 \copyright \, Year\,  Diogenes  Co., Sofia
   \par  \noindent pp. xxx--xxx, DOI: ......................
   \hfill  \vspace*{-36pt}
  }

 \title[Series representation of the pricing \dots]{Series representation of the pricing Formula for the European option driven by space-time fractional diffusion}
 \author[\normalsize J.-Ph. Aguilar, C. Coste, J. Korbel]{\normalsize Jean-Philippe Aguilar$^{1}$, Cyril Coste$^{2}$, Jan Korbel$^{3,4,5}$ }

\begin{document}

 \vbox to 2.5cm { \vfill }


 \bigskip \medskip

 \begin{abstract}

In this paper, we show that the price of an European call option, whose underlying asset price is
driven by the space-time fractional diffusion, can be expressed in terms of rapidly convergent double-series. This series formula is obtained from the Mellin-Barnes representation of the option price with help of residue summation in $\mathbb{C}^2$. We also derive the series representation for the associated risk-neutral factors, obtained by Esscher transform of the space-time fractional Green functions.
 \medskip

                  

 \smallskip


 \end{abstract}

 \maketitle

 \vspace*{-16pt}


\section{Introduction}

In the financial literature, models based on L\'evy (or $\alpha$-stable) distributions~\cite{Zolotarev86} play a prominent role, because such distributions possess heavy tails and thus allow extreme but realistic events, such as sudden jumps of market prices, 
that Gaussian models fail to describe; their relevance in financial modeling has been known since the works of Mandelbrot and Fama in the 1960s~\cite{Mandelbrot63,Fama65}. They are also closely related to fractional analysis: when the price log-returns are driven by an $\alpha$-stable distribution with maximal negative asymmetry (or skewness)~\cite{Carr03} then, after some suitable transformations, the price of an option on this asset is solution to a space-fractional PDE with boundary conditions \cite{KK16}. 

Such models recently gained in popularity, because, as noticed by Walter in his epistemological work on on financial models \cite{C13}, technology has changed our perception of risk. At the time when traders could only see a closing price on their screens (that is, the final price on a given trading day), then the Gaussian hypothesis had  a dominating influence on their mind. But when data providers became able to collect and redistribute intraday market data, it became clear that the intraday prices exhibited continuous jumps, and therefore traders and market makers started taking heavy tail models into account.


With high frequency trading, a new revolution has begun: nowadays, financial engineers need to consider the aggregation of a large number of trades in short periods of time, alternating with non-trading periods. The clock time is no longer adapted to the real market dynamics and the old hypothesis of "market time" seems to be much more suitable. Montroll and Weiss \cite{Montroll65} introduced a very simple idea: instead of considering fixed time steps, they allowed the steps to vary randomly with some statistical distribution. 
This model, called the Continous Time Random Walk Model (CTRW), is a good framework for modeling the tick by tick dynamic of financial assets: Gorenflo et al. \cite{Gorenflo00} explained why the CTRW is a statistically relevant candidate for modeling German and Italian bond prices (Bund and BTP) and derived the corresponding time fractional PDE. 

To mix the advantages of both space and time fractionality, space-time double-fractional diffusion has been introduced and extensively studied from the theoretical point of view \cite{Gorenflo99,Zatloukal14,Luchko16,Luchko16a,Mainardi07,Mainardi10,Stynes16}; however, it has only been recently considered in financial modeling \cite{zhu14,Gong,KK16,koleva,Korbel16}. 
It features a more complete structure than the simple composition of the time and space fractional models
as it exhibits non trivial phenomena including larges jumps and memory effects, which can not be understood as a simple market time re-parametrization of an $\alpha$-stable process. Let us also mention that it has also found many applications in real systems -- financial processes representing one of the most promising fields, where the fractional diffusion and generally fractional calculus has been successfully applied \cite{Akrami,funahashi,Jizba18,kleinert,Kerrs,Pagnini04,Tarasova17,Tarasova18}.



Nevertheless, when it comes to option pricing, the old Gaussian model first described by Black and Scholes \cite{Black} is still the most used by market practitioners. The main reason is that this model is analytically solvable, that is, the price of an option can be easily expressed in terms of elementary functions of the market parameters. Realistic generalizations of the Black-Scholes model such as switching multifractal models~\cite{Calvet08}, stochastic volatility models~\cite{Heston93,jizba09} or jump processes~\cite{Tankov03} possess, at best, semi-closed pricing formulas or must be solved with help of numerical simulations. And, as for the space-time fractional diffusion model we mentioned earlier, pricing formulas take the form of Mellin-Barnes integral representation~\cite{KK16}, which are intractable for practitioners, and whose numerical estimation can be erroneous and time consuming. The purpose of this paper is to show that it is possible to transform this integral representation into a rapidly converging double-series, which does not involve any advanced mathematical operators. Moreover, one can easily control the numerical precision of the resulting price. The calculation of the series is based on multidimensional Mellin transform and residue summation in $\mathbb{C}^2$.


The paper is organized as follows: Section 2 introduces basic concepts in multiple Mellin-Barnes integrals and discusses the properties of fractional diffusion fractional diffusion and its applications to option pricing. In Section 3, we derive a closed formula for the so-called risk-neutral parameter. The main result of the paper, i.e., the series representation for an European call option driven by the space-time fractional diffusion, is presented in Section 4, with discussion of several special cases. The final section is dedicated to conclusions. 

\section{Preliminary results}

In this section, we briefly summarize the main results about option pricing based on fractional diffusion. The first space-fractional option pricing model was introduced by Carr and Wu \cite{Carr03} and it has been generalized for the case of space-time fractional diffusion in Ref. \cite{KK16}. These models are strongly related to Mellin-Barnes integrals; in order to evaluate these integrals, we start by introducing some concepts in multidimensional complex analysis and residue theory.

\subsection{Mellin transform and residue summation}

We enumerate, without proof, concepts and properties of the Mellin transform in one and two dimensions, that will be useful for deriving the main results of this paper. Proofs and full details on the theory of the one-dimensional Mellin transform are provided in \cite{Flajolet95}. An introduction to multidimensional complex analysis can be found e.g., in the classic textbook \cite{Griffiths78}, and applications to the specific cases of Mellin-Barnes integrals were developed in \cite{Tsikh94,Tsikh97}.

\subsubsection{One-dimensional Mellin transform}
Let us briefly summarize the main properties of Mellin transform:
\begin{enumerate}[label=\textbf{\arabic*}.,wide, labelwidth=!, labelindent=0pt]
\item The Mellin transform of a locally continuous function $f$ defined on $\mathbb{R}^+$ is the function $f^*$ defined by
\begin{equation}\label{Mellin_def}
f^*(s) \, := \, \int\limits_0^\infty \, f(x) \, x^{s-1} \, \id x
\end{equation}
The region of the complex plane $\{ \alpha < Re (s) < \beta \}$ into which the integral \eqref{Mellin_def} converges is often called the fundamental strip of the transform, and sometimes denoted $ < \alpha , \beta  > $.

\item The Mellin transform of the exponential function is, by definition, the Euler Gamma function:
\begin{equation}
\Gamma(s) \, = \, \int\limits_0^\infty \, e^{-x} \, x^{s-1} \, \id x
\end{equation}
with strip of convergence $\{ Re(s) > 0 \}$. Outside of this strip, it can be analytically continued, except at every negative integer $s=-n$ where it admits the singular behavior
\begin{equation}\label{sing_Gamma}
\Gamma(s) \, \underset{s\rightarrow -n}{\sim} \, \frac{(-1)^n}{n!}\frac{1}{s+n} \hspace*{1cm} n\in\mathbb{N}
\end{equation}
\item The inversion of the Mellin transform is performed via an integral along any vertical line in the strip of convergence:
\begin{equation}\label{inversion}
f(x) \, = \, \int\limits_{c-i\infty}^{c+i\infty} \, f^*(s) \, x^{-s} \, \frac{\id s}{2i\pi} \hspace*{1cm} c\in ( \alpha, \beta )
\end{equation}
and notably for the exponential function one gets the so-called \textit{Cahen-Mellin integral}:
\begin{equation}\label{Cahen}
e^{-x} \, = \, \int\limits_{c-i\infty}^{c+i\infty} \, \Gamma(s) \, x^{-s} \, \frac{\id s}{2i\pi} \hspace*{1cm} c>0
\end{equation}
\item When $f^*(s)$ is a ratio of products of Gamma functions of linear arguments:
\begin{equation}
f^*(s) \, = \, \frac{\Gamma(a_1 s + b_1) \dots \Gamma(a_n s + b_n)}{\Gamma(c_1 s + d_1) \dots \Gamma(c_m s + d_m)}
\end{equation}
then one speaks of a \textit{Mellin-Barnes integral}, whose \textit{characteristic quantity} is defined to be
\begin{equation}\label{Delta_1D}
\Delta \, = \, \sum\limits_{k=1}^n \, a_k \, - \, \sum\limits_{j=1}^m \, c_j
\end{equation}
$\Delta$ governs the behavior of $f^*(s)$ when $|s|\rightarrow \infty$ and thus the possibility of computing \eqref{inversion} by summing the residues of the analytic continuation of $f^*(s)$ right or left of the convergence strip:
\begin{equation}
\left\{
\begin{aligned}
& \Delta < 0 \hspace*{1cm} f(x) \, = \, -\sum\limits_{Re(s_N) > \beta} \, \res_{S_N} f^*(s)x^{-s}  \\
& \Delta > 0 \hspace*{1cm} f(x) \, = \, \sum\limits_{Re(s_N) < \alpha} \, \res_{S_N} f^*(s)x^{-s}
\end{aligned}
\right.
\end{equation}
For instance, in the case of the Cahen-Mellin integral one has $\Delta = 1$ and therefore:
\begin{equation}
e^{-x} \, = \, \sum\limits_{Re(s_n)<0} \res_{s_n} \Gamma(s) \, x^{-s} \, = \, \sum\limits_{n=0}^{\infty} \, \frac{(-1)^n}{n!}x^n
\end{equation}
as expected from the usual Taylor series of the exponential function.
\end{enumerate}

\subsubsection{Multidimensional Mellin transform} Mellin transform can also be extended to the multidimensional domain. Below are the main properties of the multidimensional Mellin transform:
\begin{enumerate}[label=\textbf{\arabic*}.,wide, labelwidth=!, labelindent=0pt]
\item Let $\underline{a}_k$, $\underline{c}_j$, be vectors in $\mathbb{C}^2$, and $b_k$, $d_j$ some complex numbers. Let $\underline{t}:=\begin{bmatrix} t_1 \\ t_2 \end{bmatrix}$ and $\underline{c}:=\begin{bmatrix} c_1 \\ c_2 \end{bmatrix}$ in $\mathbb{C}^2$. The symbol "." denotes the Euclidean scalar product. We speak of a Mellin-Barnes integral in $\mathbb{C}^2$ when one deals with an integral of the type
\begin{equation}
\int\limits_{\underline{c}+i\mathbb{R}^2} \, \omega
\end{equation}
where $\omega$ is a complex differential 2-form which reads
\begin{equation}
\omega \, = \, \frac{\Gamma(\underline{a}_1.\underline{t}_1 + b_1) \dots \Gamma(\underline{a}_n.\underline{t}_n + b_n)}{\Gamma(\underline{c}_1.\underline{t}_1 + d_1) \dots \Gamma(\underline{c}_m.\underline{t}_m + b_m)} \, x^{-t_1} \, y^{-t_2} \, \frac{\id t_1}{2i\pi} \wedge \frac{\id t_2}{2i\pi} \quad \, x,y \in\mathbb{R}
\end{equation}
The singular sets induced by the singularities of the Gamma functions
\begin{equation}
D_k \, := \, \{ \underline{t}\in\mathbb{C}^2 \, , \, \underline{a}_k.\underline{t}_k + b_k = -n_k \, , \, n_k \in\mathbb{N}   \} \,\,\,\, \, k=0 \dots n
\end{equation}
are called the \textit{divisors} of $\omega$. The \textit{characteristic vector} of $\omega$ is defined to be
\begin{equation}
\Delta \, = \, \sum\limits_{k=1}^n \underline{a}_k \, - \, \sum\limits_{j=1}^m \underline{c}_j
\end{equation}
and the \textit{admissible half-plane}:
\begin{equation}
\Pi_\Delta \, := \, \{ \underline{t}\in\mathbb{C}^2 \, , \, Re( \Delta . \underline{t} ) \, < \, Re( \Delta . \underline{c} )\,  \}
\end{equation}
\item Let the $\rho_k$ in $\mathbb{R}$, the $h_k:\mathbb{C}\rightarrow\mathbb{C}$ be linear applications and $\Pi_k$ be a subset of $\mathbb{C}^2$ of the type
\begin{equation}\label{Pik}
\Pi_k \, := \, \{ \underline{t}\in\mathbb{C}^2, \, Re(h_k(t_k)) \, < \, \rho_k \}\, .
\end{equation}
A \textit{cone} in $\mathbb{C}^2$ is a Cartesian product
\begin{equation}
\Pi \, = \, \Pi_1 \times \Pi_2
\end{equation}
where $\Pi_1$ and $\Pi_2$ are of the type \eqref{Pik}. Its \textit{faces} $\varphi_k$ are
\begin{equation}
\varphi_k \, := \, \partial \Pi_k \hspace*{1cm} k=1,2
\end{equation}
and its \textit{distinguished boundary}, or \textit{vertex} is
\begin{equation}
\partial_0 \, \Pi \, := \, \varphi_1 \, \cap \, \varphi_2\, .
\end{equation}
\item Let $1<n_0<n$. We group the divisors $D=\cup_{k=0}^n \, D_k$ of the complex differential form $\omega$ into two sub-families
\begin{equation}
D_1 \, := \, \cup_{k=1}^{n_0} \, D_k, \,\,\, \,\,\, D_2 \, := \, \cup_{k=n_0+1}^{n} \, D_k,  \hspace*{1cm}  D \, = \, D_1\cup D_2.
\end{equation}
We say that a cone $\Pi\subset\mathbb{C}^2$ is \textit{compatible} with the divisors family $D$ if:
\begin{enumerate}
\item[-] \, Its distinguished boundary is $\underline{c}$;
\item[-] \, Every divisor $D_1$ and $D_2$ intersect at most one of his faces:
\begin{equation}
D_k \, \cap \, \varphi_k \, = \, \emptyset \hspace*{1cm} \mathrm{for} \ k=1,2.
\end{equation}
\end{enumerate}

\item Residue theorem for multidimensional Mellin-Barnes integral \cite{Tsikh94,Tsikh97}: If $\Delta \neq 0$ and if $\Pi\subset\Pi_\Delta$ is a compatible cone located into the admissible half-plane, then
\begin{equation}\label{res_thm_C2}
\int\limits_{\underline{c}+i\mathbb{R}^2} \, \omega \, = \, \sum\limits_{\underline{t}\in\Pi\cap(D_1 \cap D_2)} \res_{\underline{t}} \, \omega
\end{equation}
and the series converges absolutely. The residues are to be understood as the "natural" generalization of the Cauchy residue, that is:
\begin{multline}
\res_0 \, \left[ f(t_1,t_2) \, \frac{\id t_1}{2i\pi t_1^{\alpha_1}} \wedge \frac{\id t_2}{2i\pi t_1^{\alpha_2}}  \right] \, = \, \frac{1}{(\alpha_1-1)!(\alpha_2-1)!} \times \\ \frac{\partial ^{\alpha_1+\alpha_2-2}}{\partial t_1^{\alpha_1-1} \partial t_2^{\alpha_2-1} } f(t_1,t_2) |_{t_1=t_2=0}
\end{multline}
where $\alpha_1$ and $\alpha_2$ are strictly positive integers. 
\end{enumerate}

\subsection{Space-time fractional diffusion}
Space-time (double)-fractional diffusion equation is a generalization of the ordinary diffusion equation for non-natural derivatives. One of the most popular forms is based on Caputo time-fractional derivative and Riesz-Feller fractional derivative. It can be expressed as
\begin{equation}\label{double_fractional}
\left({}^\ast_0 \mathcal{D}^\gamma_t -  \mu [{}^\theta \mathcal{D}^\alpha_x] \right) g(x,t) = 0\,
\end{equation}
where $x \in \mathds{R}$ and $t \in [0,\infty)$. Parameters can acquire the following values:  $\alpha \in (0,2]$, $\gamma \in (0,\alpha]$. Asymmetry parameter $\theta$ is defined in the so-called \emph{Feller-Takayasu diamond} $|\theta| \leq \min \left\{\alpha, 2-\alpha \right\}$.  ${}^\ast_0 \mathcal{D}^\gamma_t$ denotes the \emph{Caputo fractional derivative}, which is defined as
\begin{equation}
{}^\ast_{t_0} \mathcal{D}^\nu_t f(t) =
\frac{1}{\Gamma(\lceil \nu \rceil - \nu)} \int_{t_0}^t \frac{f^{\lceil \nu \rceil}(\tau)}{(t - \tau)^{ \nu +1-\lceil \nu \rceil}} \ud \tau \end{equation}
and ${}^\theta \mathcal{D}^\alpha_x$ denotes the \emph{Riesz-Feller fractional derivative}, which is usually defined via its Fourier image as

\begin{equation}
\mathcal{F}[{}^\theta \mathcal{D}^\nu_x
f(x)](k) = - {}^\theta \psi^\nu (k)F[f(x)](k) = - \mu |k|^\nu e^{i(\mathrm{sign} k) \theta \pi /2} \mathcal{F}[f(x)](k)\, .
\end{equation}
Naturally, both derivatives become ordinary derivative operators for the order of the derivative is  a natural number. According to the order of temporal-derivative $\gamma$, the equation requires one or two conditions. Apart from standard initial condition
\begin{equation}
g(x,t=0) = f_0(x)\, ,
\end{equation}
it is for $\gamma >1$ necessary to input the condition
\begin{equation}
\frac{\partial g(x,t)}{\partial t} |_{t=0} = f_1(x)\, .
\end{equation}
Typically, we choose $f_1(x) \equiv 0$ (this is also used in the rest of this paper).

The space-time fractional diffusion has been studied by many authors, perhaps the most famous is the paper by Gorenflo et al.\cite{Gorenflo99}, where it is also possible to find all technical details. Here, we briefly revise several important aspects of the space-time fractional diffusion. First, the scaling of the fundamental solution (also called \emph{Green function}) $g(x,t)$ is given by the scaling exponent $\Omega$, so we obtain the scaling
\begin{equation}
g(x,t) = \frac{1}{t^\Omega} \, G \left(\frac{x}{t^\Omega}\right)
\end{equation}
where $\Omega = \gamma/\alpha$. Second, for particular values of parameters, it is possible to recover well-known distributions. For $\gamma=1$, i.e., space-fractional diffusion, we recover L\'{e}vy diffusion driven by $\alpha$-stable distribution. For $\gamma=1$ and $\alpha=2$ we recover normal diffusion driven by Gaussian distribution.

In order to express the solution $g(x,t)$, it is usual to transform Eq.~\eqref{double_fractional} into its Fourier-Laplace image ($x \stackrel{\mathcal{F}}{\rightarrow} k$, $t \stackrel{\mathcal{L}}{\rightarrow} s$). Let us remind the initial conditions $f_1(x) = 0$ and $f_0(x) = \delta(x)$. This leads to the algebraic equation

\begin{equation}\label{lapfur}
\hat{\bar{g}}^\theta_{\alpha,\gamma}(k,s) s^\gamma -s^{\gamma-1} + {}^\theta \psi^\alpha(k) \hat{\bar{g}}^\theta_{\alpha,\gamma}(k,s) = 0\, .
\end{equation}

The original solution $g_{\alpha,\gamma}^\theta(x,t)$ can be expressed by the inverse Fourier-Laplace transform. We show two important representations of the fundamental solution of Eq.~\eqref{double_fractional}.

The first representation, in detail discussed in Ref.~\cite{Zatloukal14}, is important mainly for the case $\gamma <1$ and is based on \emph{Schwinger trick} $1/A = \int_0^\infty e^{- l A} \ud l$. This enables to rewrite the expression $1/(s^\gamma + {}^\theta \psi^\alpha(k))$ as $\int_0^\infty e^{-l s^\gamma} e^{- l {}^\theta \psi^\alpha(k))} \ud l$ so it is possible to separate functions depending of $s$ and $k$.  Consequently, it is possible to rewrite the distribution as an integral composition of two kernels

\begin{equation}\label{smearing}
g_{\alpha,\gamma}^\theta(x,t) = \int_0^\infty g_\gamma(t,l) g_\alpha^\theta(l,x)\, \ud l
\end{equation}
where $g_\gamma$ and $g_{\alpha}^\theta$ are solutions of single-fractional diffusion equations
\begin{eqnarray}
\frac{\partial g_\gamma(t,l)}{\partial l} &=& {}^\ast_0 \mathcal{D}^\gamma_t \, g_\gamma(t,l)\, ,\\
\frac{\partial g_\alpha^\theta(l,x)}{\partial l} &=&
{}^\theta \mathcal{D}^\alpha_x \, g_\alpha^\theta(l,x)\, .
\end{eqnarray}

Each equation represents one class of single-fractional diffusion processes. The first equation describes time-fractional diffusion, while the second represents the space-fractional diffusion leading to $\alpha$-stable distributions. It is formally possible to use this representation also for $\gamma >1$, but in this case is $g_\gamma(t,l)$ not anymore positive and therefore cannot be interpreted as a \emph{smearing kernel} (details can be found in \cite{Zatloukal14,KK16}). On the other hand, it can be useful to use the Schwinger representation for calculation of some derived quantities, as e.g., moments of the distribution. We demonstrate this approach in Section \ref{sec:RN}, where we use this representation in order to calculate the risk-neutral factor corresponding to the space-time fractional diffusion.

Alternatively, Eq. \eqref{lapfur} can be solved by Mellin transform technique resulting into the Mellin-Barnes integral. The inverse Laplace transform of Eq.~\eqref{lapfur} reads~\cite{Gorenflo99}:
\begin{equation}
\hat{g}(t,k) = E_{\gamma}({}^\theta \psi^\alpha(k) t^\gamma)\, ,
\end{equation}
where $E_\gamma(x) = \sum_{n=0}^\infty \frac{x^n}{\Gamma(\gamma n + 1)}$ is the Mittag-Leffler function. It is possible to represent it via the Mellin-Barnes integral as
\begin{equation}
E_a(z) = \frac{1}{2 \pi i } \int\limits_{c - i \infty}^{c+i \infty} \frac{\Gamma(t_1)\Gamma(1-t_1)}{\Gamma(1-a t_1)} (-z)^{-t_1} \ud t_1\, ,
\end{equation}
where $c \in (0,1)$, which is given by the Mellin transform theorem \cite{Flajolet95}. After plugging back and straightforward calculation, one ends with Mellin-Barnes representation of space-time fractional Green function as

\begin{eqnarray}\label{Green_function_DF}
g_{\alpha,\gamma}^\theta(x,t) = \frac{ 1 }{\alpha x}\frac{1}{2 \pi i}   \int\limits_{c_1 - i \infty}^{c_1 + i \infty}
\Gamma \left[
         \begin{array}{ccc}
           \frac{t_1}{\alpha} & 1-\frac{t_1}{\alpha} & 1-t_1 \\
           1-\frac{ \gamma t_1}{\alpha} & \frac{\alpha-\theta}{2\alpha}t_1 & 1 - \frac{\alpha-\theta}{2\alpha}t_1 \\
         \end{array}
       \right]\nonumber \\
       \times \left(\frac{x}{(-\mu t^{\gamma})^{1/\alpha}}\right)^{t_1} \, \ud t_1
\end{eqnarray}
where \emph{Gamma fraction} is defined as $\Gamma\left[\begin{array}{ccc}
                                                         x_1 & \dots & x_n \\
                                                         y_1 & \dots & y_m
                                                       \end{array}\right] = \frac{\Gamma(x_1)\dots\Gamma(x_n)}{\Gamma(y_1)\dots\Gamma(y_m)}$.

\section{Price of European call option in the space-time fractional model}

We recall the principles of option pricing and their applications to double fractional diffusion models. At the end of this section, we notably focus on series representation of the risk-neutral factor, which generalizes the well-known Black-Scholes risk-neutral factor $\frac{\sigma^2}{2}$ to our wider class of models.

\subsection{Option pricing}

Option pricing consists in two important aspects. First, a realistic model of underlying price and second, an appropriate hedging policy which maximally eliminates the risk. Optimally, the risk should be completely eliminated. Based on the assumption of \emph{efficient market}, the most popular hedging policy is the \emph{risk-neutral} pricing. In this scenario, the option price of European type, i.e., option with given maturity time, is calculated as
\begin{equation}\label{expect}
C(S_t,K,r,\sigma,t) = e^{-r \tau}\langle C(S_T,K,r,\sigma,T) | \mathcal{F}_t \rangle_\mathbb{Q}
\end{equation}
where $\tau=T-t$. 
$\mathbb{Q}$ is the so-called \emph{risk-neutral measure}, equivalent to original probability measure $\mathbb{P}$ describing the price evolution. 
For exponential processes described by its log-returns, the risk-neutral measure is given by \emph{Esscher transform} \cite{Gerber93}. The terminal condition at $t=T$ (or equivalently, the initial condition for $\tau=0$) is given by the option's payoff, which, for a call option, is equal to
\begin{equation}
 C(S_T,K,r,\sigma,T) \, = \, \max \{ S_T - K , 0    \} \, =: \, [S_T-K]^+\, .
\end{equation}

For the space-time fractional model, the call option price \eqref{expect} can be expressed as the convolution of the Green function \eqref{Green_function_DF} and the payoff (after some suitable change of variables) \cite{KK16}:
\begin{equation}\label{propagator}
C_{\alpha,\gamma}^\theta (S,K,r,\mu,\tau) = e^{-r \tau} \int\limits_{-\infty}^\infty  \ \left[S e^{\tau(r-q+\mu)+y} - K\right]^+ g_{\alpha,\gamma}^\theta(y,\tau) \ud y\, .
\end{equation}
Note that in our future calculations we will take, without loss of generality, a dividend $q=0$ in order to simplify the notations. Factor $\mu$ appearing the ``modified payoff'' is a result of the risk-neutral measure $\mathbb{Q}$, which is obtained by the Esscher transform of the original measure $\mathbb{P}$. Details can be found in~\cite{KK16}. It is possible to calculate this so-called \emph{risk-neutral factor} $\mu$ as
\begin{equation}
\mu = - \log \int e^y g_{\alpha,\gamma}^\theta(y,\tau=1) \ud y\,
\end{equation}
when the integral exists. Obviously, the necessary condition of integral convergence is exponential decay in positive tail of the probability distribution. This can be assured as soon as the \textit{maximal (negative) asymmetry condition} holds, that is when $\theta = \alpha-2$, $\alpha>1$. This fully asymmetric case was for space-fractional diffusion discussed in Ref.~\cite{Carr03}, for space-time fractional diffusion in Refs.~\cite{KK16,Korbel16}.

Let us briefly discuss the interpretation of main model parameters, i.e., the derivative orders $\alpha$, $\gamma$ and $\sigma$ in option pricing. First, parameter $\sigma$ has the role of scale parameter and can be interpreted as the market risk. This means that if $\sigma$ increases, uncertainty in the market increases and\emph{all} option prices also increase and vice versa. On the other hand, this is not the case for parameters $\alpha$ and $\gamma$. As discussed in \cite{KK16,Korbel16}, they play the role of \emph{risk redistribution} parameters. The role of $\alpha$ characterizes spatial redistribution, because with decreasing $\alpha$ the negative tail of the distribution becomes ``heavier'' (the decay is slower) and the probability of large drops increases dramatically. This has been extensively discussed in \cite{Carr03}. Similarly, parameter $\gamma$ characterizes ``temporal'' risk redistribution, caused e.g., by some memory effects \cite{Tarasova18,Teyssiere}. Thus, it increases the risk for short/long-term options, while decreases the risk for the other type. The presence of space-time risk redistribution has the impact on various observable phenomena, e.g., to the shape of volatility-smile \cite{AK18}.




\subsection{Risk-neutral factor for the space-time fractional diffusion}\label{sec:RN}

 It is unfortunately not always possible to express the Risk-neutral factor $\mu$ analytically. Nevertheless in the spatial-diffusion case (that is for $\gamma=1$) , it is known that ~\cite{Carr03}
\begin{equation}\label{mu_Carr_Wu}
\mu_1 = \left( \frac{\sigma}{\sqrt{2}}\right) ^\alpha \sec{\frac{\pi \alpha}{2}}
\end{equation}
when the maximal negative asymmetry assumption is fulfilled. For space-time fractional case it is possible to derive an integral representation based on Eq.~\eqref{smearing} and to to rewrite $\mu$ as
\begin{eqnarray}\label{muDF1}
\mu &=& - \log \int\limits_{-\infty}^\infty \, \ud y \, e^y \int\limits_0^\infty \, \ud l \,  g_\gamma(\tau=1,l) g_\alpha^\theta(l,x) \nonumber \\
&=& - \log \int\limits_0^\infty \, \ud l g_\gamma(\tau = 1,l) e^{-\left(\frac{\sigma}{\sqrt{2}} l \right)^\alpha \sec \left(\frac{\pi \alpha}{2}\right) } \, .
\end{eqnarray}
The last representation was obtained by change of integrals. For a Caputo time fractional derivatives, it can be shown \cite{KK16,Gorenflo99} that for $\gamma < 1$ 
\begin{equation}
g_\gamma(\tau,l) \, = \, \frac{1}{\tau^\gamma} \, M_\gamma\left( \frac{l}{\tau^\gamma} \right)
\end{equation}
where $M_\nu(z)$ is a function of Wright type, which admits the following Mellin-Barnes representation \cite{Mainardi10}:
\begin{equation}
M_\nu(z) \,  =  \, \int\limits_{c-i\infty}^{c+i\infty} \, \frac{\Gamma(s)}{\Gamma(\nu s + 1 - \nu) } \, z^{-s} \, \frac{\ud s}{2i\pi} \hspace*{1cm} c>0\, .
\end{equation}
Interestingly, the Mellin-Barnes representation is also valid for the case $\gamma >1$, but it does not lead to a positive smearing kernel. Plugging into \eqref{muDF1} and recalling \eqref{mu_Carr_Wu} we obtain:
\begin{equation}\label{muDF2}
\mu \, = \, - \log \int\limits_{c-i\infty}^{c+i\infty} \, \frac{\Gamma(s)}{\Gamma(\gamma s + 1 - \gamma) } \, \int\limits_0^\infty l^{-s} e^{-l^\alpha \, \mu_1} \, \ud l \, \frac{\ud s}{2i\pi}\, .
\end{equation}
The integral over $l$ is straightforward to perform, because it is the integral representation of the Gamma function \cite{Abramowitz72}:
\begin{equation}\label{mugamma}
\int\limits_0^\infty l^{-s} \, e^{-l^\alpha \mu_1} \, \ud l = \, \frac{1}{\alpha}\Gamma\left(\frac{1-s}{\alpha}\right) \, \mu_1^{\frac{s-1}{\alpha}}\, .
\end{equation}
The integral converges in the strip $Re(s)<1$. As a result, we can formulate a proposition, which is a simple consequence of \eqref{muDF2} and \eqref{mugamma}:
\begin{proposition}
Let $\sigma>0$, $1<\alpha \leq 2$, and  $\mu_1=\left( \frac{\sigma}{\sqrt{2}}\right) ^\alpha \sec{\frac{\pi \alpha}{2}}$, then the risk-neutral factor $\mu$ admits the representation:
\begin{equation}\label{mu_MB}
\mu \, = \, - \log \left[ \frac{1}{\alpha} \, \int\limits_{c-i\infty}^{c+i\infty} \, \frac{\Gamma(s)\Gamma(\frac{1-s}{\alpha})}{\Gamma(\gamma s + 1 - \gamma) } \, \mu_1^{\frac{s-1}{\alpha}} \, \frac{\ud s}{2i\pi} \right] \hspace*{1cm} 0 < c < 1\, .
\end{equation}
\end{proposition}

\noindent
Let us now express \eqref{mu_MB} in the series representation, which can be more convenient for numerical applications:
\begin{proposition}
Let $\sigma>0$ and $1<\alpha \leq 2$, and $\mu_1=\left( \frac{\sigma}{\sqrt{2}}\right) ^\alpha \sec{\frac{\pi \alpha}{2}}$, then for any $\gamma > 1 - \frac{1}{\alpha}$ the risk-neutral factor $\mu$ can be expressed in the form of the absolutely convergent series:
\begin{equation}\label{mu_series}
\mu \, = \, - \log \, \sum\limits_{n=0}^{\infty} \, \frac{(-1)^n \Gamma(1+\alpha n)}{n!\Gamma(1+\gamma\alpha n)} \mu_1^n\, .
\end{equation}
\end{proposition}
\begin{proof}
The characteristic quantity $\Delta$ (see \eqref{Delta_1D}) associated to the Mellin-Barnes integral \eqref{mu_MB}, which governs its decay at infinity, is equal to
\begin{equation}
\Delta \, = \, 1 - \frac{1}{\alpha} - \gamma \, < \, 0
\end{equation}
and therefore the line-integral in \eqref{mu_MB} is equal to minus the sum of the residues located in the right half plane $\{Re ( s) > 1 \}$.
\begin{figure}[t]
\centering
\includegraphics[scale=0.5]{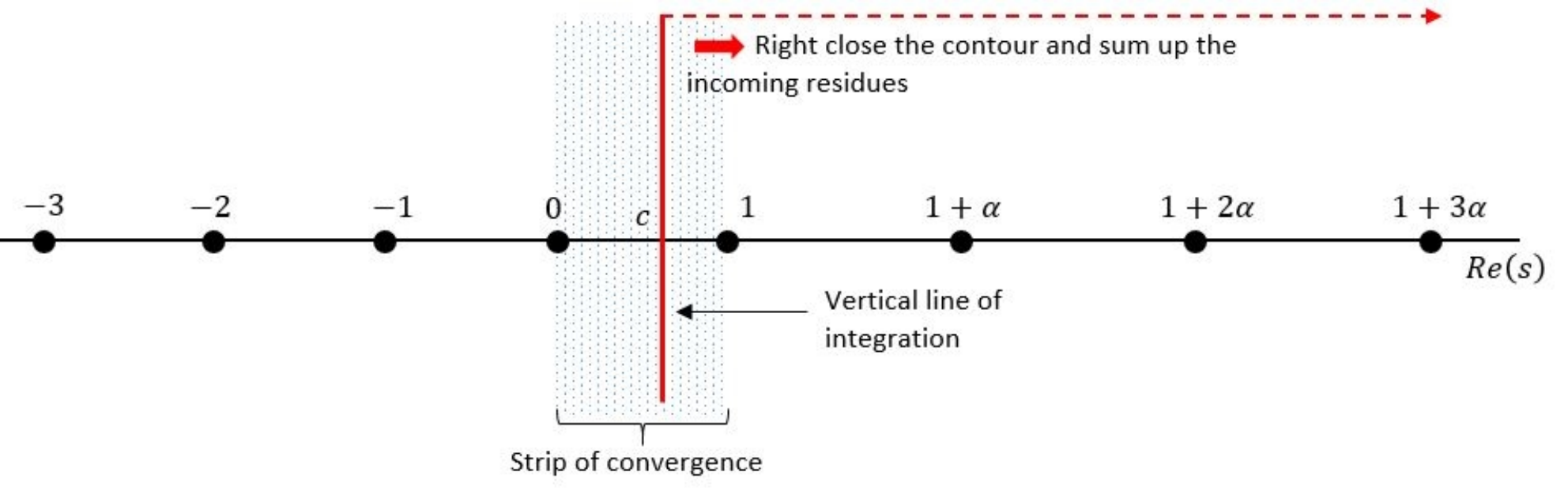}
\caption{Singularities for the Mellin-Barnes integral \eqref{mu_MB}. The poles located left of the convergence strip in $s=0,-1,-2,\dots$ are induced by the $\Gamma(s)$ term, those located right of the convergence strip in $s=1,1+\alpha,1+2\alpha,\dots$ are induced by the $\Gamma(\frac{1-s}{\alpha})$ term.}
\label{fig1}
\end{figure}
They are induced by the poles of the $\Gamma(\frac{1-s}{\alpha})$ function, which arise at every negative argument, that is when $s=1+\alpha n, \, n \in\mathbb{N}$. Around theses points, the $\Gamma(\frac{1-s}{\alpha})$ function admits the singular behavior (see \eqref{sing_Gamma}):
\begin{equation}
\Gamma\left(\frac{1-s}{\alpha}\right) \, \underset{s \rightarrow 1+\alpha n}{\sim} \, \frac{(-1)^n}{n!}\frac{\alpha}{-s+1+\alpha n}
\end{equation}
and therefore the residues associated to the Mellin-Barnes integral in \eqref{mu_MB} are:
\begin{equation}
\res (s = 1+\alpha n) \, = \, -\frac{(-1)^n\Gamma(1+\alpha n)}{n!\Gamma(\gamma(1+\alpha n) + 1 - \gamma)} \, \mu_1^n\, .
\end{equation}
Simplifying, taking minus the sum of this residues and the overall logarithm yields the formula \eqref{mu_series}
\end{proof}

\noindent
We may note that, when taking $\gamma = 1$ in formula \eqref{mu_series}, we are left with
\begin{equation}
\mu \, =  \, - \log \, \sum\limits_{n=0}^{\infty} \frac{(-1)^n\mu_1^n}{n!} \, = \, - \log \, e^{-\mu_1} \, = \, \mu_1
\end{equation}
as expected. Moreover, it follows from the classical Taylor expansion for $\log(1+x)$ that:
\begin{equation}
\mu \, = \, \frac{\Gamma(1+\alpha)}{\Gamma(1+\gamma\alpha)}\mu_1 \, + \, O\left(\mu_1^2\right)
\end{equation}
and, in the case $\alpha=2$, we have:
\begin{equation}\label{mu_series_BS}
\mu \, = \, -\frac{\sigma^2}{\Gamma(1+2\gamma)}  \, + \, O\left(\sigma^4\right)
\end{equation}
which resumes to the Gaussian parameter $-\frac{\sigma^2}{2}$ when $\gamma = 1$.




\section{Series representation of the European call option}

Let us now turn the attention to the series representation of Eq. \eqref{propagator}. This will be done in two steps. First, we use Mellin-Barnes representation for the Green function corresponding to space-time fractional solution. Second, we use the residue summation formula in order to obtain the double-series representation.

Let us assume that $S,K,r,\tau \, > \, 0$. Let $1< \alpha \leq 2$ and $0<\gamma \leq \alpha$. We will denote the vectors $\underline{Z}\in\mathbb{C}^2$ by:
\begin{equation}
\underline{Z} \, :=
\begin{bmatrix}
Z_1 \\ Z_2
\end{bmatrix}
\hspace*{1cm} Z_1,\,Z_2 \in \mathbb{C}
\end{equation}
We will assume that the Carr-Wu maximal negative asymmetry hypothesis $\theta = \alpha - 2$ holds, and we will denote the call price by $C_{\alpha,\gamma}^{\alpha - 2}(S,K,r,\mu,\tau):=C_{\alpha,\gamma}(S,K,r,\mu,\tau)$.

\subsection{Mellin-Barnes representation for the call price}

 We first derive a representation for the call price \eqref{propagator} under the form of a complex integral in $\mathbb{C}^2$.
\begin{proposition}
Let $[\log]:=\log\frac{S}{K} + r\tau$ and let $P\subset\mathbb{C}^2$ be the polyhedra
\begin{equation}
P \, := \, \{ \underline{t}\in\mathbb{C}^2 \, , \, 0 < Re (t_2) < 1, \, Re(t_2-t_1)>1 \}
\end{equation}
Then, for any $\underline{c}\in P$, the following holds:
\begin{eqnarray}\label{Call_C2}
C_{\alpha,\gamma} (S,K,r,\mu,\tau) =&& \nonumber\\
\frac{K e^{-r \tau}}{\alpha} \int\limits_{c_1-i\infty}^{c_1+i\infty} \int\limits_{c_2-i\infty}^{c_2+i\infty}  \, (-1)^{-t_2}
\frac{\Gamma(t_2)\Gamma(1-t_2)\Gamma(-1-t_1+t_2)}{\Gamma(1-\frac{\gamma}{\alpha}t_1)} \,&& \nonumber\\
\times
(-[\log]-\mu\tau)^{1+t_1-t_2}(-\mu\tau^{\gamma})^{-\frac{t_1}{\alpha}}\frac{\ud t_1}{2i\pi}\wedge\frac{\ud t_2}{2i\pi}\, .&&
\end{eqnarray}
\end{proposition}
\begin{proof}
With maximal asymmetry hypothesis, the Green function $g_{\alpha,\gamma}^{\alpha - 2} (y,\tau) \, := \, g_{\alpha,\gamma}(y,\tau)$ simplifies into (see eq. \eqref{Green_function_DF}):
\begin{equation}
g_{\alpha,\gamma}(y,\tau) \, = \, \frac{1}{\alpha y} \int\limits_{c_1-i\infty}^{c_1+i\infty} \, \frac{\Gamma(1-t_1)}{\Gamma(1-\frac{\gamma}{\alpha}t_1)} \, \left( \frac{y}{(-\mu\tau^{\gamma})^{\frac{1}{\alpha}} } \right)^{t_1} \, \frac{\ud t_1}{2i\pi} \, 
\end{equation}
for $0 < c_1 < 1$. Inserting this into formula \eqref{propagator} yields:
\begin{multline}\label{propagator_2}
C_{\alpha,\gamma} (S,K,r,\mu,\tau) = \frac{K e^{-r \tau}}{\alpha} \,
 \int\limits_{c_1-i\infty}^{c_1+i\infty} \, \frac{\Gamma(1-t_1)}{\Gamma(1-\frac{\gamma}{\alpha}t_1)} \, \\
  \times  \int\limits_{-[\log]-\mu\tau}^{\infty} (e^{[\log]+\mu\tau + y}-1)y^{t_1-1} \, \ud y (-\mu\tau^{\gamma})^{-\frac{t_1}{\alpha}}\frac{\ud t_1}{2i\pi}\, .
\end{multline}
Here we have used the fact that $[Se^{(r+\mu)\tau +y}-K]^+ = K[e^{[\log]+\mu\tau+y}-1]^+ $. It is possible to integrate by parts in \eqref{propagator_2}, with the result:
\begin{multline}\label{propagator_3}
C_{\alpha,\gamma} (S,K,r,\mu,\tau) = -\frac{K e^{-r \tau}}{\alpha} \,
 \int\limits_{c_1-i\infty}^{c_1+i\infty} \, \frac{\Gamma(1-t_1)}{\Gamma(1-\frac{\gamma}{\alpha}t_1)} \, \frac{1}{t_1} \, \\ \times \int\limits_{-[\log]-\mu\tau}^{\infty} e^{[\log]+\mu\tau + y} \, y^{t_1} \, \ud y (-\mu\tau^{\gamma})^{-\frac{t_1}{\alpha}}\frac{\ud t_1}{2i\pi}\, .
\end{multline}
Let us introduce the following Mellin-Barnes representation of the exponential term (see Eq. \eqref{Cahen}):
\begin{equation}
e^{[\log]+\mu\tau + y} \, = \, \int\limits_{c_2-i\infty}^{c_2+i\infty} \, (-1)^{-t_2} \Gamma(t_2) ([\log] + \mu\tau + y )^{-t_2} \, \frac{dt_2}{2i\pi}\, ,
\end{equation}
for $c_2 > 0$. Plugging into \eqref{propagator_3} transforms the integral over the Green variable $y$ into a \textit{Beta integral} \cite{Abramowitz72}, with the result:
\begin{eqnarray}
 \int\limits_{-[\log]-\mu\tau}^{\infty} \, ([\log] + \mu\tau + y )^{-t_2} y^{t_1} \, \ud y \, \nonumber\\
 = \, (-[\log]-\mu\tau)^{1+t_1-t_2} \frac{\Gamma(1-t_2)\Gamma(-1-t_1+t_2)}{\Gamma(-t_1)}
\end{eqnarray}
Replacing in \eqref{propagator_3}, using the functional relation $-t_1\Gamma(-t_1) = \Gamma(1-t_1)$ and simplifying the fraction, we are left with the double integral in $\mathbb{C}^2$, as shown in Eq. \eqref{Call_C2}.
The integral converges when the arguments of the Gamma functions in the numerator are positive, that is whenever $t_2>0$, $t_2<1$ and $-1-t_1+t_2>0$.
\end{proof}

\noindent
The integral formula \eqref{Call_C2} can be expressed as a sum of residues in some region of $\mathbb{C}^2$, which is shown in the next section.

\subsection{Residue summation}

\begin{theorem}
Let $1 < \alpha \leq 2$ and $1-\frac{1}{\alpha} < \gamma \leq \alpha$. Under maximal negative asymmetry hypothesis (i.e., $\theta=\alpha-2$), the European call price driven by space-time fractional diffusion is:
\begin{equation}\label{Formula}
C_{\alpha,\gamma}(S,K,r,\mu,\tau) \, = \, \frac{Ke^{-r\tau}}{\alpha} \, \sum\limits_{\substack{n = 0 \\ m = 1}}^{\infty} \, \frac{(-1)^n}{n!\Gamma(1-\gamma\frac{n-m}{\alpha})} (-[\log]-\mu\tau)^{n}(-\mu\tau^{\gamma})^{\frac{m-n}{\alpha}}\, .
\end{equation}
\end{theorem}
\begin{proof}
Let $\omega$ be the complex differential 2-form
\begin{multline}\label{Call_C2_2}
\omega \, : = \, (-1)^{-t_2}
\frac{\Gamma(t_2)\Gamma(1-t_2)\Gamma(-1-t_1+t_2)}{\Gamma(1-\frac{\gamma}{\alpha}t_1)}\, \\
\times
(-[\log]-\mu\tau)^{1+t_1-t_2}(-\mu\tau^{\gamma})^{-\frac{t_1}{\alpha}}\frac{\ud t_1}{2i\pi}\wedge\frac{\ud t_2}{2i\pi}
\end{multline}
so that the call price \eqref{Call_C2} can be written under the compact form
\begin{equation}
C_{\alpha,\gamma}(S,K,r,\mu,\tau) \, = \, \frac{Ke^{-r\tau}}{\alpha} \int\limits_{\underline{c}+i\mathbb{R}^2} \, \omega\, .
\end{equation}
The characteristic vector associated to $\omega$ is (see \cite{Tsikh94,Tsikh97}) :
\begin{equation}
\Delta \, = \,
\begin{bmatrix}
-1 + \frac{\gamma}{\alpha} \\ 1
\end{bmatrix}\, .
\end{equation}
Therefore, the half-plane of convergence one must considers:
\begin{equation}
\Pi_\Delta \, := \, \, \, \, \, \, \{ \underline{t} \in \mathbb{C}^2, \, Re( \Delta \, . \, \underline{t}) \,  \, < \,  \, Re( \Delta \, . \, \underline{c} \, ) \}
\end{equation}
is the one located under the line (see Fig. \ref{fig2}):
\begin{equation}
Re(t_2) \, = \, (1-\frac{\gamma}{\alpha}) (Re(t_1) \, - \, c_1) + c_2\, .
\end{equation}
\begin{figure}[t]
\centering
\includegraphics[scale=0.4]{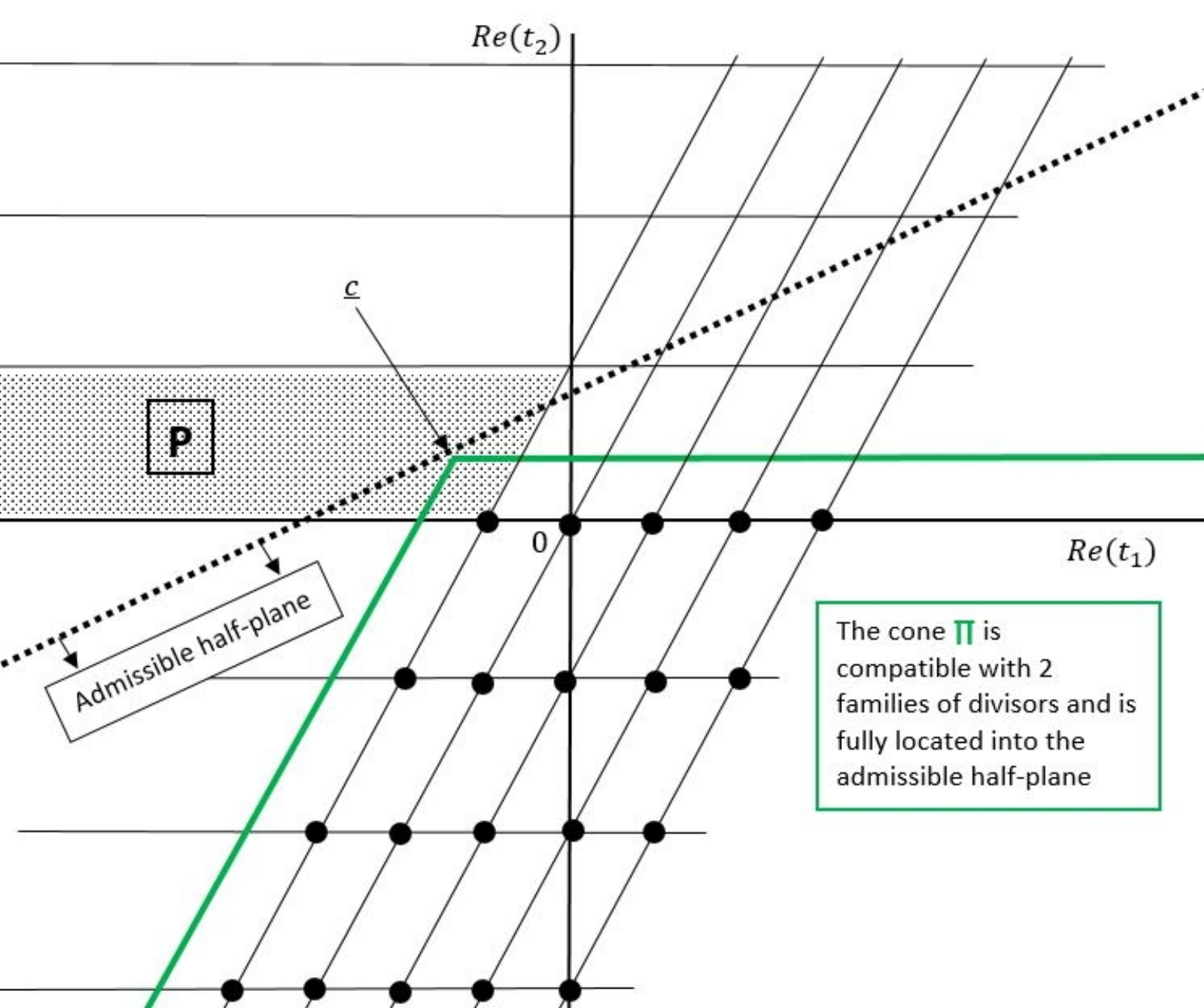}
\caption{The admissible region $\Pi_{\Delta}$, for the complex 2-form $\omega$, is the one located under the dotted oblique line. There is only one compatible cone in this region: the green cone, which is compatible with the two family of divisors $D_1$ (oblique lines) and $D_2$ (horizontal lines). The sum of the residues at every $\mathbb{C}^2$-singularity (points) into this cone is therefore equal to the integral of $\omega$.}
\label{fig2}
\end{figure}

\noindent Because $\gamma\leq\alpha$, we have $0 < 1-\frac{\gamma}{\alpha} \leq 1$ and therefore the cone $\Pi$ defined by
\begin{equation}
\Pi \, := \,\,\,\, \{\underline{t} \in \mathbb{C}^2 \, , \, Re(t_2) < 0 \, , \, Re(-t_1+t_2) < 1 \}
\end{equation}
is included in $\Pi_\Delta$. Moreover, it is compatible with the two families of divisors
\begin{equation}
\left\{
\begin{aligned}
& D_1 \, = \, \left\{\underline{t}\in\mathbb{C}^2,  -1 - t_1 + t_2 = -n_1 \,\, , \,\,\, n_1 \in\mathbb{N} \right\}
\\
& D_2 \, = \, \left\{\underline{t}\in\mathbb{C}^2, t_2 = -n_2 \,\, , \,\,\, n_2 \in\mathbb{N} \right\}
\end{aligned}
\right.
\end{equation}
induced by the $\Gamma(-1-t_1+t_2)$ and $\Gamma(t_2)$ functions respectively. To compute the residues associated to every element of the singular set $D := D_1 \cap D_2$, we change the variables:
\begin{equation}
\left\{
\begin{aligned}
& u_1 \, := \, -1 -t_1 + t_2 \\
& u_2 \, := \, t_2
\end{aligned}
\right.
\longrightarrow
\left\{
\begin{aligned}
& t_1 \, = \, -1+u_2-u_1 \\
& t_2 \, = \, u_2 
\end{aligned}
\right.
\end{equation}
so that in this new configuration $\omega$ reads
\begin{multline}
\omega \, = \, (-1)^{-u_2} \, \\ \frac{\Gamma(u_2)\Gamma(1-u_2)\Gamma(u_1)}{\Gamma(1-\gamma \frac{-1+u_2-u_1}{\alpha})} \left(-[\log]-\mu\tau \right)^{-u_1} (-\mu\tau^\gamma)^{\frac{1+u_1-u_2}{\alpha}} \, \frac{\id u_1}{2i\pi} \wedge \frac{\id u_2}{2i\pi}
\end{multline}
With this new variables, the divisors $D_1$ and $D_2$ are induced by the $\Gamma(u_1)$ and $\Gamma(u_2)$ functions, and intersect at every point of the type $(u_1,u_2)=(-n,-m)$, $n,m\in\mathbb{N}$. From the singular behavior of the Gamma function \eqref{sing_Gamma} around a singularity, we can write:
\begin{multline}
\omega \, \underset{(u_1,u_2)\rightarrow (-n,-m)} {\sim} \,  \frac{(-1)^{n+m}}{n!m!} 
(-1)^{-u_2} \, \frac{\Gamma(1-u_2)}{\Gamma(1-\gamma \frac{-1+ u_2-u_1}{\alpha})} \\
\times \left(-[\log]-\mu\tau \right)^{-u_1} (-\mu\tau^\gamma)^{\frac{1+u_1-u_2}{\alpha}} \, \frac{\id u_1}{2i\pi (u_1+n)} \wedge \frac{\id u_2}{2i\pi (u_2 + m)}
\end{multline}
Taking the residues and simplifying:
\begin{equation}
\res_{(-n,-m)} \, \omega \, = \,  \frac{(-1)^n}{n!\Gamma(1-\gamma\frac{n-m-1}{\alpha})} \left(-[\log]-\mu\tau \right)^{n} (-\mu\tau^\gamma)^{\frac{1+m-n}{\alpha}}\, 
\end{equation}
From the residue theorem for Mellin-Barnes integrals in $\mathbb{C}^2$ (see equation \eqref{res_thm_C2}), we know that the integral \eqref{Call_C2_2} is equal to the sum of residues into the whole cone $\Pi$:
\begin{multline}
C_{\alpha,\gamma}(S,K,r,\mu,\tau) \, = \, \\ \frac{Ke^{-r\tau}}{\alpha} \, \sum\limits_{\substack{n = 0 \\ m = 0}}^{\infty} \, \frac{(-1)^n}{n!\Gamma(1-\gamma\frac{n-m-1}{\alpha})} (-[\log]-\mu\tau)^{n}  (-\mu\tau^{\gamma})^{\frac{1+m-n}{\alpha}}\, 
\end{multline}
Performing the index substitution $m\rightarrow m+1$ yields the representation \eqref{Formula} and completes the proof.
\end{proof}

\subsection{Special cases}

Let us discuss several special parameter choices corresponding to well-known diffusion models: 
\begin{itemize}
\item \underline{Space-fractional diffusion:} When $\gamma=1$ in the series \eqref{Formula}, the series expansion is simplified and corresponds to the call price under the so-called Finite Moment L\'evy Stable model \cite{Carr03}
\begin{equation}
C_{\alpha,1} \, = \, \frac{Ke^{-r\tau}}{\alpha} \, \sum\limits_{\substack{n = 0 \\ m = 1}}^{\infty} \, \frac{(-1)^n}{n!\Gamma(1-\frac{n-m}{\alpha})} (-[\log]-\mu\tau)^{n}(-\mu\tau)^{\frac{m-n}{\alpha}}
\end{equation}
where $\mu = \mu_1 \, = \, (\sigma/\sqrt{2})^\alpha \sec\frac{\pi\alpha}{2}$. When, additionally, $\alpha=2$, we get the series expansion for the Black-Scholes (BS) price 
\begin{equation}
C_{2,1} \, = \, \frac{Ke^{-r\tau}}{2} \, \sum\limits_{\substack{n = 0 \\ m = 1}}^{\infty} \, \frac{(-1)^n}{n!\Gamma(1-\frac{n-m}{2})} \left(-[\log]+\frac{\sigma^2}{2}\tau \right)^{n} \left(\frac{\sigma^2}{2}\tau \right)^{\frac{m-n}{2}}
\end{equation}

\item \underline{Neural diffusion:}
For $\alpha=\gamma$, the diffusion equation becomes a generalization of the wave equation \cite{luchko13} (obtained for $\alpha=2$). In this case, the ratio $\gamma/\alpha=1$, and therefore the formula can be expressed as
\begin{equation}
C_{\alpha,\alpha} = \frac{Ke^{-r\tau}}{\alpha} \sum\limits_{\substack{n = 0 \\ m = 1 \\ m\geq n}}^{\infty} \, \frac{(-1)^n}{n! (m-n)!} (-[\log]-\mu\tau)^{n}((-\mu)^{1/\alpha} \tau)^{m-n}
\end{equation}
Let us note that the neural diffusion is not typical for financial processes and it is more important from the theoretical point of view, because it represents the borderline case of fractional diffusion.

\item \underline{Time-fractional diffusion:} Taking $\alpha=2$ in \eqref{Formula} yields the series expansion for the time-fractional BS price:
\begin{equation}\label{BS_frac}
C_{2,\gamma} \, = \, \frac{Ke^{-r\tau}}{2} \, \sum\limits_{\substack{n = 0 \\ m = 1}}^{\infty} \, \frac{(-1)^n}{n!\Gamma(1-\gamma\frac{n-m}{2})} (-[\log]-\mu\tau)^{n}(-\mu\tau^{\gamma})^{\frac{m-n}{2}}
\end{equation}

\item \underline{At-the-money forward approximation of time-fractional diffusion:} Assuming that the asset is at-the-money forward, that is:
\begin{equation}
S \, = \, Ke^{-r\tau}
\end{equation}
then, by definition, $[\log]=0$ and the fractional BS price \eqref{BS_frac} becomes:
\begin{equation}\label{BS_frac_ATMF}
C_{2,\gamma}^{(ATMF)} \, = \, \frac{S}{2} \, \sum\limits_{\substack{n = 0 \\ m = 1}}^{\infty} \, \frac{(-1)^n}{n!\Gamma\left(1-\gamma\frac{n-m}{2}\right)} (-\mu)^{\frac{n+m}{2}}\tau^{\frac{(2-\gamma)n+\gamma m}{2}}
\end{equation}
As $\gamma<2$, the series \eqref{Call_C2_2} is a power series (which is not the case in the general series \eqref{BS_frac}, where negative powers arise) which starts for $n=0,m=1$ and goes as:
\begin{equation}
C_{2,\gamma}^{(ATMF)}\, = \, \frac{S}{2} \, \left[  \frac{\sigma}{\Gamma(1+\frac{\gamma}{2})} \, \sqrt{\frac{\tau^\gamma}{\Gamma(1+2\gamma)}}  \, + \, \mathcal{O}(\sigma^2)  \right]
\end{equation}
where we have used the approximation~\eqref{mu_series_BS} for the risk-neutral parameter. Taking $\gamma=1$ and recalling $\Gamma(\frac{3}{2})=\frac{\sqrt{\pi}}{2}$, we are left with
\begin{equation}
C_{2,1}^{(ATMF)}  \, = \, \frac{S}{\sqrt{2\pi}} \, \sigma\sqrt{\tau} \, + \mathcal{O}(\sigma^2) \, \simeq \, 0.4 S \sigma\sqrt{\tau}
\end{equation}
which the well-known Brenner-Subrahmanyam approximation for the BS price \cite{BS94}.
\end{itemize}


\subsection{Convergence of the double-sum representation}
In order to demonstrate the speed of convergence, let us calculate the contributions of each term in the double-sum \eqref{Formula} for a typical option price. Table \ref{fig:series} provides an example of the series convergence of an option with realistic parameters $S=3800, K=4000, r=1\%, \sigma= 20\%, \tau = 1, \alpha = 1.7, \gamma =0.9$; we observe that the convergence is very fast. We see that for the numerical precision of three decimal places, it is only necessary to sum up to $n=6$ and $m=6$. 

\begin{table}[h!]
\centering
\begin{scriptsize}
\begin{tabular}{|c||ccccccc|}
  \hline
 {\quad \bfseries n \textbackslash \ m }  &  1 & 2 & 3 & 4 & 5 & 6 & 7   \\
  \hline
  \hline
  0 & 429.751 & 60.850 & 7.216 & 0.749 & 0.070 & 0.006 & 0.000     \\
  1 & -203.666 & -37.572 & -5.320 & -0.6315 & -0.065 & -0.006 & -0.000   \\
  2 & 28.893 & 8.903 & 1.642 & 0.233 & 0.0.028 & 0.003 & 0.000    \\
  3 & 0.549 & -0.842 & -0.259 & -0.048 & -0.007 & -0.000 & -0.000   \\
  4 & -0.352 & -0.012 & 0.018 & 0.006 & 0.001 & 0.000 & 0.000   \\
  5 & -0.016 & 0.006 & 0.000 & -0.000 & -0.000 & -0.000 & -0.000  \\
  6 & 0.005 & 0.000 & -0.000 & -0.000 & 0.000 & 0.000 & 0.000   \\
  7 & 0.000 & -0.000 & -0.000 & 0.000 & 0.000 & -0.000 & -0.000   \\
  \hline
  Call & 255.162 & 286.495 & 289.792 & 290.090 & 290.126 & 290.128 & 290.128    \\
  \hline
\end{tabular}
\end{scriptsize}
\caption{Table containing the numerical values for the $(n,m)$-term in the series (\ref{Formula}) for the option price ($S=3800, \, K=4000, \, r=1\%, \sigma=20\%, \, \tau=1Y, \, \alpha=1.7$,$\gamma=0.9$). The call price converges to a precision of $10^{-3}$ after summing only very few terms of the series.}
\label{fig:series}
\end{table}

\section{Conclusions}
In this paper, we have introduced a new representation for the European option driven by a space-time fractional diffusion equation in the form of rapidly convergent double series \eqref{Formula}. This double series can be derived from the Mellin-Barnes integral representation of the European option with help of residue summation in $\mathbb{C}^2$. The series representation of the double-fractional option pricing model, which incorporates risk redistribution in both spatial and temporal domain, might be useful in real trading, since the formula can be easily used by practitioners without any deeper knowledge of advanced mathematical techniques (such as Mellin transform or residue summation in multidimensional complex analysis), and is an explicit function of observable market parameters. Moreover, it is possible to control the numerical precision of the pricing formula. Contrary to other representations, no numerical technique is needed to evaluate the option, which would typically be the case for complicated integral representations, where the integrals cannot be expressed analytically. 

Interestingly, the residue summation technique can be used in more applications, as shown for the case of the risk-neutral factor of the space-time fractional Green function obtained by the Esscher transform. One can think about more applications, as expressions for optimal hedging policies, optimal exercise times for American options, etc. One could 
even go beyond the field of financial processes and use the residue summation techniques to calculate general functions of random random variables driven by space-time fractional diffusion, or more generally, by any process, whose Green function can be expressed by means of Mellin-Barnes integral representation. 

\section*{Acknowledgements}
J. K. acknowledges support from the Austrian Science Fund, Grant No. I 3073, and from the Czech Science Foundation, Grant No. 17–33812L. 

ˇ









 \bigskip \smallskip

 \it

 \noindent
$^1$ BRED Banque Populaire, Modeling Department, 18 quai de la R\^{a}p\'{e}e, Paris - 75012, FRANCE\\[4pt]
  e-mail: jean-philippe.aguilar@bred.fr\\[4pt]
\hfill Received:  \\[12pt]
$^2$  MAIF, 200 avenue Salvador Allende, Niort, FRANCE \\[4pt]
  e-mail: cyril.coste@maif.fr\\[12pt]
$^3$  Section for the Science of Complex Systems, CeMSIIS, Medical
University of Vienna, Spitalgasse 23, A-1090, Vienna, AUSTRIA \\[4pt]
$^4$ Complexity Science Hub Vienna, Josefst\"{a}dterstrasse 39, 1080 Vienna, AUSTRIA \\[4pt]
$^5$ Faculty of Nuclear Sciences and Physical Engineering, Czech Technical University in Prague, B\v{r}ehov\'{a} 7, 115 19, Prague, CZECH REPUBLIC \\[4pt]
  e-mail: jan.korbel@meduniwien.ac.at

\end{document}